\documentclass[envcountsame]{llncs}
\usepackage{amsmath,amssymb}
\usepackage{hyperref}
\usepackage{cleveref}
\usepackage{enumerate}
\usepackage{versions}
\usepackage{mathtools}
\usepackage[maxcitenames=2,maxbibnames=9,backend=bibtex,doi=false,firstinits=true,isbn=false,url=false]{biblatex}
\bibliography{references}

\includeversion{SHORT}
\excludeversion{LONG}

\usepackage{mathtools} 

\newtheorem{observation}{Observation}

\newcommand{\bigo}[1]{\mathcal{O}(#1)}
\newcommand{\bigolr}[1]{\mathcal{O}\left(#1\right)}

\newcommand{\set}[1]{$\mathcal{#1}$}

\begin{document}

\hyphenation{pro-blems}
\pagestyle{headings} 
\addtocmark{Computing Coverage Kernels Under Restricted Settings}

\mainmatter 
\title{Computing Coverage Kernels  \\  Under Restricted Settings}

\author{
 J\'er\'emy Barbay \inst{1}
\thanks{%
	This work was supported by projects CONICYT Fondecyt/Regular nos 1170366 and 1160543, and CONICYT-PCHA/Doctorado Nacional/2013-63130209 (Chile).
}
\and
Pablo P\'erez-Lantero \inst{2}
 \and 
 Javiel Rojas-Ledesma\inst{1}
}

\institute{
 Departamento de Ciencias de la Computaci\'on, Universidad de Chile, Chile\\ 
 \email{jeremy@barbay.cl, jrojas@dcc.uchile.cl.}
\and
Departamento de Matem\'atica y  Ciencia de la Computaci\'on, \\ 
Universidad de Santiago, Chile.
\email{pablo.perez.l@usach.cl.}
}



\maketitle

\setcounter{footnote}{0}

\begin{abstract} 
We consider the \textsc{Minimum Coverage Kernel} problem: given a set \set{B} of $d$-dimensional boxes, find a subset of \set{B} of minimum size covering the same region as \set{B}.  This problem is $\mathsf{NP}$-hard, 
but as for many $\mathsf{NP}$-hard problems on graphs, the problem becomes solvable in polynomial time
under restrictions on the graph induced by $\mathcal{B}$. 
We consider various classes of graphs, show that \textsc{Minimum Coverage Kernel} remains $\mathsf{NP}$-hard even for severely  restricted  instances, and provide two polynomial time approximation algorithms for this problem.
\end{abstract}


\section{Introduction}\label{sec:Introduction}				 

Given a set  $P$ of $n$ points, and a set \set{B} of $m$ boxes (i.e. axis-aligned closed hyper-rectangles) in $d$-dimensional space, the \textsc{Box Cover} problem consists in finding a set $\mathcal{C \subseteq B}$  of minimum size such that \set{C} covers $P$.
A special case is the \textsc{Orthogonal Polygon Covering} problem: given an orthogonal polygon \set{P} with $n$ edges, find a set of boxes $\mathcal{C}$ of minimum size whose union covers \set{P}.
Both problems are \textsf{NP}-hard~\cite{CulbersonR94,Fowler1981}, but their known approximabilities in polynomial time are different: 
while \textsc{Box Cover} can be approximated up to a factor within $\bigo{\log \mathtt{OPT}}$, where $\mathtt{OPT}$ is the size of an optimal solution~\cite{BronnimannG95,Clarkson2007}; 
\textsc{Orthogonal Polygon Covering} can be approximated up to a factor within $\bigo{\sqrt{\log n}}$~\cite{KumarR03}.
 In an attempt to better understand what makes these problems hard, and why there is such a gap in their approximabilities, we introduce the notion of coverage kernels and study its computational complexity.
 
Given a set \set{B} of $n$ $d$-dimensional boxes,
a \emph{coverage kernel} of \set{B} is a subset $\mathcal{K} \subseteq \mathcal{B}$ covering the same region as \set{B},
and a minimum coverage kernel of \set{B} is a coverage kernel of minimum size.
%
%
%
The computation of  a minimum coverage kernel (namely, the \textsc{Minimum Coverage Kernel} problem) is intermediate between the \textsc{Orthogonal Polygon Covering} and the \textsc{Box Cover} problems.
This problem has found applications (under distinct names, and slight variations) in the compression of access control lists in networks~\cite{DalyLT16},
and in obtaining concise descriptions of structured sets in databases~\cite{LakshmananNWZJ02,PuM05}.
Since \textsc{Orthogonal Polygon Covering} is $\mathsf{NP}$-hard, the same holds for the 
\textsc{Minimum Coverage Kernel} problem.
%
%
We are interested in the exact computation and approximability of \textsc{Minimum Coverage Kernel} in various restricted settings:
\vskip -10pt
\begin{enumerate}
	\item  \textbf{Under which restrictions is the exact computation of \textsc{Minimum Coverage Kernel} still $\mathsf{NP}$-hard}? 
	\item \textbf{How precisely can one approximate a \textsc{Minimum Coverage Kernel}} in polynomial time?
\end{enumerate}

When the interactions between the boxes in a set $\mathcal{B}$ are simple (e.g., when all the boxes are disjoint), a minimum coverage kernel of $\mathcal{B}$ can be computed efficiently.
A natural way to capture the complexity of these interactions is through the intersection graph.
%
The intersection graph of $\mathcal{B}$ is the un-directed graph with a vertex for each box,
and in which two vertices are adjacent if and only the respective boxes intersect.
When the intersection graph is a tree, for instance, 
each box of $\mathcal{B}$ is either completely covered by another, or present in any coverage kernel of $\mathcal{B}$, 
and thus a minimum coverage kernel can be computed efficiently.
For problem on graphs, a common approach to understand when does an \textsc{NP}-hard problem become easy  is to study distinct restricted classes of graphs, in the hope to define some form of ``boundary classes'' of inputs separating ``easy'' from ``hard'' instances~\cite{AlekseevBKL07}.
Based on this, 
we study the hardness of the problem under restricted classes of the intersection graph of the input.



\paragraph{Our results.}
We study the \textsc{Minimum Coverage Kernel} problem under three restrictions of the intersection graph,
commonly considered for other problems~\cite{AlekseevBKL07}:
planarity of the graph, bounded clique-number, and bounded vertex-degree.
We show that the problem remains $\mathsf{NP}$-hard even when the intersection graph of the boxes has clique-number at most 4, and the maximum degree is at most 8. 
For the \textsc{Box Cover} problem we show that it remains $\mathsf{NP}$-hard even under the severely restricted setting where the intersection graph of the boxes is planar, its clique-number is at most 2 (i.e., the graph is triangle-free), the maximum degree is at most 3, and every point is contained in at most two boxes.

We complement these hardness results with two approximation algorithms for the \textsc{Minimum Coverage Kernel} problem running in polynomial time.
We describe a $\bigo{\log n}$-approximation algorithm which runs in time within $\bigo{\mathtt{OPT} \cdot n^{\frac{d}{2} + 1}\log^2 n}$;
and a randomized algorithm computing a $\bigo{\log \mathtt{OPT}}$-approximation in expected time within $\bigo{\mathtt{OPT} \cdot n^{\frac{d+1}{2}} \log^2 n}$, with high probability (at least $1- \frac{1}{n^{\Omega(1)}}$).
Our main contribution in this matter is not the existence of polynomial time approximation algorithms (which can be inferred from results on \textsc{Box Cover}), but a new data structure which allows to significantly improve the running time of finding those approximations (when compared to the approximation algorithms for \textsc{Box Cover}).
This is relevant in applications where a minimum coverage kernel needs to be computed repeatedly~\cite{Agarwal2014,DalyLT16,LakshmananNWZJ02,PuM05}.

In the next section we review the reductions between the three problems we consider, and introduce some basic concepts.
We then present the hardness results in \Cref{sec:covkernels_hardness}, and describe in \Cref{sec:covkernels_approximation} the two approximation algorithms.
We conclude in \Cref{sec:cover_discussion} with a discussion on the results and future work.

\section{Preliminaries}~\label{sec:background}
\begin{figure}[t]
	\begin{center}
		\includegraphics[page=2,scale=1]{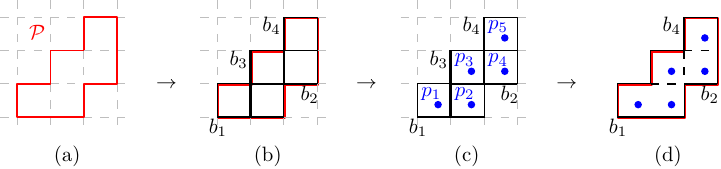}
	\end{center}   
	\caption{
		 a) An orthogonal polygon $\mathcal{P}$. 
		 b) A set of boxes $\mathcal{B} = \{ b_1, b_2, b_3, b_4\}$ covering exactly $\mathcal{P}$, and such that in any cover of $\mathcal{P}$ with boxes, every box is either in $\mathcal{B}$, or fully covered by a box in $\mathcal{B}$.
		 c) A set of points $\mathcal{D(B)} = \{p_1, p_2, p_3, p_4, p_5\}$ such that any subset of $\mathcal{B}$ covering $\mathcal{D(B)}$, covers also $\mathcal{P}$.
		 d) The subset $\{b_1, b_2, b_4\}$ is an optimal solution for the \textsc{Orthogonal Polygon Cover} problem on $\mathcal{P}$, the \textsc{Minimum Coverage Kernel} problem on $\mathcal{B}$, and the \textsc{Box Cover} problem on $\mathcal{D(B)}, \mathcal{B}$.
	}\label{fig:intro}
\end{figure}
To better understand the relation between the \textsc{Orthogonal Polygon Covering}, the \textsc{Box Cover} and the \textsc{Minimum Coverage Kernel} problems,  we briefly review the reductions between them.
We describe them in the Cartesian plane, as the generalization to higher dimensions is straightforward.

Let $\mathcal{P}$ be an orthogonal polygon with $n$ horizontal/vertical edges.
Consider the grid formed by drawing infinitely long lines through each edge of $\mathcal{P}$ (see \Cref{fig:intro}.a for an illustration),
%
and let $G$ be the set of $\bigo{n^2}$ points of this grid lying on the intersection of two lines.
Create a set $\mathcal{B}$ of boxes as follows:
for each pair of points in $G$, if the box having those two points as opposed vertices is completely inside $\mathcal{P}$, then add it to $\mathcal{B}$ (see \Cref{fig:intro}.b.) 
Let $\mathcal{C}$ be any set of boxes covering $\mathcal{P}$. 
Note that for any box $c \in \mathcal{C}$, either the vertices of $c$ are in $G$, or $c$ can be extended horizontally and/or vertically (keeping $c$ inside $\mathcal{P}$) until this property is met.
Hence, 
there is at least one box in $\mathcal{B}$ that covers each $c \in \mathcal{C}$, respectively, 
and thus there is a subset $\mathcal{B}' \subseteq \mathcal{B}$ covering $\mathcal{P}$ with $|\mathcal{B}'| \le |C|$.
Therefore, any minimum coverage kernel of $\mathcal{B}$ is also an optimal covering of $\mathcal{P}$
(and thus, transferring the \textsf{NP}-hardness of the \textsc{Orthogonal Polygon Covering} problem~\cite{CulbersonR94} to the \textsc{Minimum Coverage Kernel} problem).

Now, let $\mathcal{B}$ be a set of $n$ boxes, and consider the grid formed by drawing infinite lines through the edges of each box in $\mathcal{B}$.
This grid has within $\bigo{n^2}$ cells ($\bigo{n^{d}}$ when generalized to $d$ dimensions).
Create a point-set $\mathcal{D(B)}$ as follows: for each cell $c$ which is completely inside a box in $\mathcal{B}$ we add to  $\mathcal{D(B)}$ the middle point of $c$  (see \Cref{fig:intro}.c for an illustration). 
We call such a point-set a \emph{coverage discretization} of $\mathcal{B}$, and denote it as $\mathcal{D(B)}$.
Note that a set $\mathcal{C} \subseteq \mathcal{B}$ covers $\mathcal{D(B)}$ if and only if $\mathcal{C}$ covers the same region as $\mathcal{B}$ (namely, $\mathcal{C}$ is a coverage kernel of $\mathcal{B}$). 
Therefore,  the \textsc{Minimum Coverage Kernel} problem is a special case of the \textsc{Box Cover} problem.

The relation between the \textsc{Box Cover}  and the \textsc{Minimum Coverage Kernel} problems has two main implications.
Firstly, hardness results for the \textsc{Minimum Coverage Kernel} problem can be transferred to the \textsc{Box Cover} problem. 
In fact, we do this in \Cref{sec:covkernels_hardness}, where we show that 
\textsc{Minimum Coverage Kernel} remains \textsf{NP}-hard under severely restricted settings, and extend this result to the \textsc{Box Cover} problem under even more restricted settings.
The other main implication is that polynomial-time approximation algorithms for the \textsc{Box Cover} problem can also be used for \textsc{Minimum Coverage Kernel}.
However, in scenarios where the boxes in $\mathcal{B}$ represent high dimensional data~\cite{DalyLT16,LakshmananNWZJ02,PuM05} and \textsc{Coverage Kernels} need to be computed repeatedly~\cite{Agarwal2014}, using approximation algorithms for \textsc{Box Cover} can be unpractical.
This is because constructing $\mathcal{D(B)}$ requires time and space within $\Theta({n^{d}})$.
We deal with this in \Cref{sec:covkernels_approximation}, where we introduce a data structure to index $\mathcal{D(B)}$ without constructing it explicitly.
Then, we show how to improve two existing approximation algorithms~\cite{BronnimannG95,Lovasz75}  for the \textsc{Box Cover} problem by using this index, making possible to use them for the \textsc{Minimum Coverage Kernel} problem in the scenarios commented on.


\section{Hardness under Restricted Settings}~\label{sec:covkernels_hardness}
We prove that \textsc{Minimum Coverage Kernel} remains $\mathsf{NP}$-hard for restricted classes of the intersection graph of the input set of boxes.
We consider three main restrictions: when the graph is planar, when the size of its largest clique (namely the clique-width of the graph) is bounded by a constant, and when the degree of a vertex with maximum degree (namely the vertex-degree of the graph) is bounded by a constant.

\subsection{Hardness of Minimum Coverage Kernel}
Consider the \textsc{$k$-Coverage Kernel} problem: given a set \set{B} of $n$ boxes, find whether there are $k$ boxes in \set{B} covering the same region as the entire set. 
Proving that \textsc{$k$-Coverage Kernel} is $\mathsf{NP}$-complete under restricted settings yields the $\mathsf{NP}$-hardness of \textsc{Minimum Coverage Kernel} under the same conditions.
To prove that \textsc{$k$-Coverage Kernel} is $\mathsf{NP}$-hard under restricted settings we reduce instances of the \textsc{Planar 3-SAT} problem (a classical $\mathsf{NP}$-complete problem~\cite{MulzerR08})
to restricted instances of \textsc{$k$-Coverage Kernel}. 
In the \textsc{Planar 3-SAT}  problem, given a boolean formula in 3-CNF whose incidence graph%
\footnote{
	The \emph{incidence graph} of a 3-SAT formula is a bipartite graph with a vertex for each variable and each clause, and an edge between a variable vertex and a clause vertex for each occurrence of a variable in a clause.
}
is planar, the goal is to find whether there is an assignment which satisfies the formula.
The (planar) incidence graph of any planar 3-SAT formula $\varphi$ can be represented in the plane as illustrated in \Cref{fig:planar3SAT-sample} for an example, where all variables lie on a horizontal line, and all clauses are represented by {\em non-intersecting} three-legged combs~\cite{KnuthR92}. We refer to such a representation of $\varphi$ as the \emph{planar embedding} of $\varphi$. 
\begin{figure}[h]
	\centering	
	\begin{minipage}{.9\textwidth}
		\centering	
		\includegraphics[scale=1]{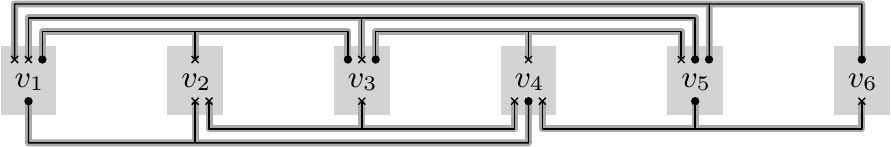}		
		\caption{\small{Planar embedding of the formula $\varphi=(v_1 \lor \overline{v_2} \lor v_3) 
				\land (v_3 \lor \overline{v_4} \lor \overline{v_5})\land(\overline{v_1} \lor \overline{v_3} \lor v_5)
				\land(v_1 \lor \overline{v_2} \lor v_4)\land(\overline{v_2} \lor \overline{v_3} \lor \overline{v_4})
				\land(\overline{v_4} \lor v_5 \lor \overline{v_6})\land(\overline{v_1} \lor v_5 \lor v_6)$.
				The crosses and dots at the end of the clause legs indicate that the connected variable appears in the clause negated or not, respectively.}}
		\label{fig:planar3SAT-sample}
	\end{minipage}
\end{figure}
Based on this planar embedding we proof the results in \Cref{theo:coverage_hardness}.
Although our arguments are described in two dimensions, they extend trivially to higher dimensions.

\begin{theorem}\label{theo:coverage_hardness}
	Let $\mathcal{B}$ be a set of $n$ boxes in the plane and let $G$ be the intersection graph of $\mathcal{B}$.
	Solving \textsc{$k$-Coverage Kernel} over $\mathcal{B}$ is \emph{\textsf{NP}}-complete even 
	if $G$ has clique-number at most 4, and vertex-degree at most 8.
\end{theorem}

\begin{proof}
	Given any set \set{B} of $n$ boxes in $\mathbb{R}^d$, and any subset  \set{K} of \set{B}, certifying that 
	\set{K} covers the same region as \set{B} can be done in  time within $\bigo{n^{d/2}}$ using \citeauthor{Chan2013}'s algorithm~\cite{Chan2013} for computing the volume of the union of the boxes in $\mathcal{B}$. Therefore, \textsc{$k$-Coverage Kernel} is in $\mathsf{NP}$.
	To prove that it is \textsf{NP}-complete, given a planar 3-SAT formula $\varphi$ with $n$ variables and $m$ clauses,  we construct a set \set{B} of $\bigo{n + m}$ boxes with a coverage kernel of size $31m + 3n$  if and only if there is an assignment of the variables satisfying $\varphi$.
	We use the planar embedding of $\varphi$ as a start point, and replace the components corresponding to variables and clauses, respectively, by gadgets composed of several boxes. 
	We show that this construction can be obtained in polynomial time, and thus any polynomial time solution to $k$-\textsc{Coverage Kernel} yields a polynomial time solution for \textsc{Planar  3-SAT}.
	We replace the components in that embedding corresponding to variables and clauses, respectively, by gadgets composed of several boxes, adding a total number of boxes polynomial in the number of variable and clauses.
	
	\begin{figure}[h]
		\centering	
		\begin{minipage}{.9\textwidth}
			\centering	
			\hspace*{-.2in}
			\includegraphics[scale=.8,page=8]{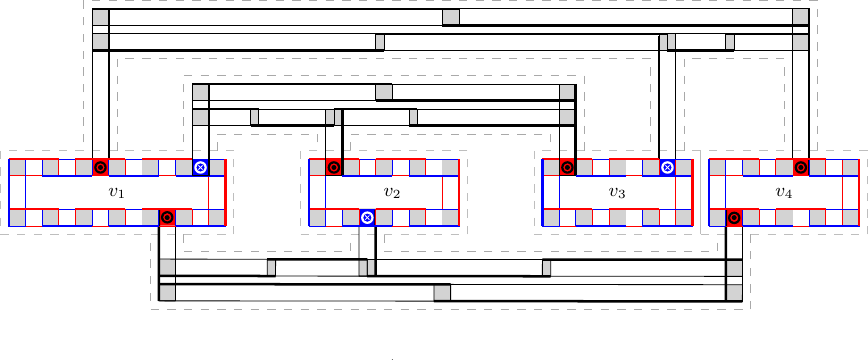}		
			\caption{Variable and clause gadgets for $\varphi = (\overline{v_1} \lor v_2 \lor v_3) \land (v_1 \lor \overline{v_2} \lor v_4) \land (v_1 \lor \overline{v_3} \lor v_4)$.
				The bold lines highlight one side of each rectangle in the instance, while the dashed lines delimit the regions of the variable and clause components in the planar embedding of $\varphi$.
				Finding a minimum subset of rectangles covering the non-white regions yields an answer for the satisfiability of $\varphi$.}
			\label{fig:clausesample-full}
		\end{minipage}
	\end{figure}

	\begin{figure}[h]
		\centering	
		\begin{minipage}{.9\textwidth}
			\centering	
			\hspace*{-.3in}
			\includegraphics[scale=.8,page=1]{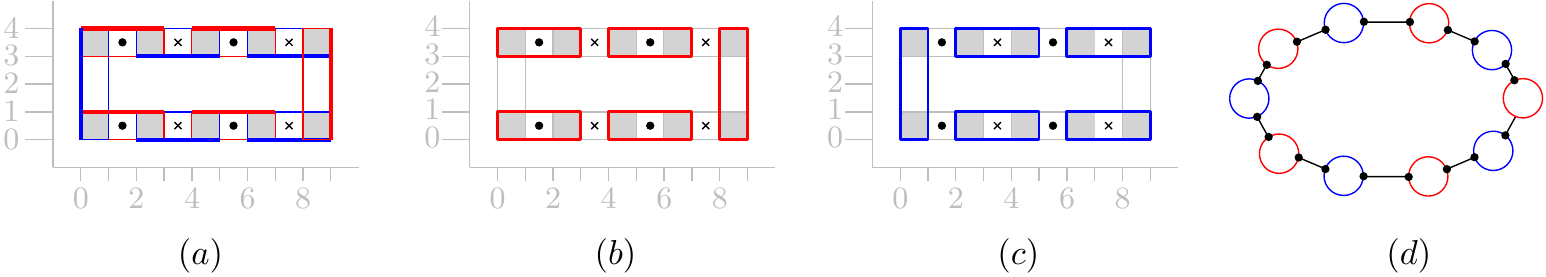}		
			\caption{
				$a)$ Gadget for a variable appearing in two clauses.
				A clause gadget connects with the variable in the regions marked with a dot or a cross, depending on the sign of the variable in the clause.
				The optimal way to cover all the redundant regions (the gray regions) is to choose either all the red rectangles (as in $b$) or all the blue rectangles (as in $c$). $d)$ The intersection graph of the variable gadget.
			}
			\label{fig:variablegadget}
		\end{minipage}
	\end{figure}
	
	\paragraph*{Variable gadgets.}
	Let $v$ be a variable of a planar 3-SAT formula $\varphi$, and let $c_v$ be the number of clauses of $\varphi$ in which $v$ appears.
	The gadget for $v$  is composed of $4c_v + 2$ rectangles colored either red or blue (see \Cref{fig:variablegadget} for an illustration): $4 c_v$ \emph{horizontal}  rectangles (of $3 \times 1$  units of size), separated into two ``rows'' with $2c_v$ rectangles each,  and two \emph{vertical} rectangles (of $1 \times 4$  units of size) connecting the rows. 
	The rectangles in each row are enumerated from left to right, starting by one. The $i$-th rectangle of the $j$-th row is defined by the product of intervals $[4(i-1), 4i-1] \times [4(j-1), 4(j-1)+1]$, for all $i=[1..2c_v]$ and $j=1,2$.
	The gadget occupies a rectangular region of $(4c_v + 1) \times 4$ units. 
	Although the gadget is defined with respect to the origin of coordinates, it is later translated to the region corresponding to $v$ in the embedding of \textcite{KnuthR92}, which we assume without loss of generality to be large enough to fit the gadget.
	Every horizontal rectangle is colored red if its numbering is odd, and blue otherwise.
	Besides, the vertical leftmost (resp. rightmost) rectangle is colored blue (resp. red). As we will see later, these colors are useful when connecting a clause gadget with its variables.
	
	Observe that:
	($i$.) every red (resp. blue) rectangle intersects exactly two others,  both blue (resp. red), sharing with each a squared region of $1 \times 1$ units (which we call \emph{redundant regions});
	($ii$.) the optimal way to cover the redundant regions is by choosing either all the $2c_v + 1$ red rectangles or all the $2c_v + 1$ blue rectangles (see \Cref{fig:variablegadget} for an example).
	
	\paragraph*{Clause gadgets.}
	Let $C$ be a clause with variables $u$, $v$, and $w$, appearing in this order from left to right in the embedding of $\varphi$.
	Assume, without loss of generality, that the component for $C$ in the embedding is above the variables. 
	We create a gadget for $C$ composed of 9 black rectangles, located and enumerated as in \Cref{fig:clausegadget}.$a$. The vertical rectangles numbered 1, 2 and 3 correspond to the legs of $C$ in the embedding, and connect with the gadgets of $u, v$, and $w$, respectively.
	The remaining six horizontal rectangles connect the three legs between them. The vertical rectangles have one unit of width and their height is given by the height of the respective legs in the embedding of $C$. Similarly, the horizontal rectangles have one unit of height and their width is given by the separation between the legs in the embedding of $C$ (see 	\Cref{fig:clausesample-full} for an example of how these rectangles are extended or stretched as needed).
	Note that:
	($i$.) every rectangle in the gadget intersects exactly two others (again, we call \emph{redundant regions} the regions where they meet); 
	($ii$.)  any minimum cover of the redundant regions (edges in \Cref{fig:clausegadget}.$b$) has five rectangles, one of which must be a leg; and 
	($iii.$)  any cover of the redundant regions which includes the three legs must have at least six rectangles (e.g., see \Cref{fig:clausegadget}.$c$). 
	
	\begin{figure}[t]
		\centering	
		\begin{minipage}{.9\textwidth}
			\centering	
			\hspace*{-.1in}
			\includegraphics[scale=.7,page=7]{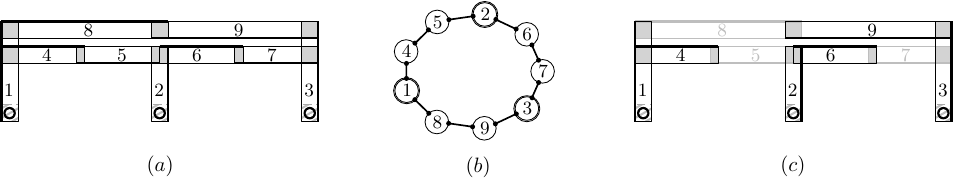}		
			\caption{
				$a)$ A clause gadget with nine black rectangles: three vertical legs (1,2, and 3) and six horizontal rectangles (4-9).
				We call the regions where they meet \emph{redundant regions}. 
				The striped regions at the bottom of each leg connect with the variables.
				$b)$ The intersection graph of the gadget. Any minimum cover of the edges (redundant regions) requires 5 vertices (rectangles).
				$c)$ Any cover of the redundant regions which includes the three legs has 6 or more rectangles.
			}
			\label{fig:clausegadget}
		\end{minipage}
	\end{figure}
	
	\paragraph{Connecting the gadgets.} Let $v$ be a variable of a formula $\varphi$ and $c_v$ be the number of clauses in which $v$ occurs. 
	The legs of the $c_v$ clause gadgets are connected with the gadget for $v$, from left to right, in the same order they appear in the  embedding of $\varphi$.
	Let $C$ be the gadget for a clause containing $v$ whose component in the embedding of $\varphi$ is above (resp. below) that for $v$.
	$C$ connects with the gadget for $v$ in one of  the rectangles in the  upper (resp. lower) row, sharing a region of $1 \times 1$  units with one of the red (resp. blue) rectangles if the variable appears positive (resp. negative) in the clause  (see \Cref{fig:clausevarsample}.$a$).
	We call this region where the variable and clause gadgets meet as \emph{connection region}, and its color is given by the color of the respective rectangle in the variable gadget.
	Note that a variable gadget has enough connection regions for all the $c_v$ clauses in which it appears, because each row of the gadget has $c_v$ rectangles of each color.
	
	\begin{figure}[h]
		\centering	
		\begin{minipage}{.9\textwidth}
			\centering	
			\hspace*{-.25in}
			\includegraphics[scale=.75,page=5]{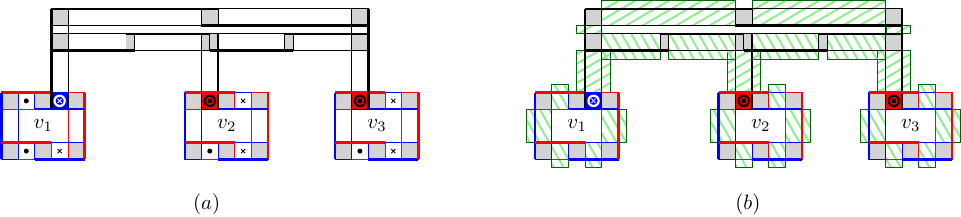}		
			\caption{An instance for $(\overline{v_1} \lor v_2 \lor v_3)$.
				$a)$ The clause gadget connects with a variable gadget in one of its red or blue connection regions depending on the sign of the variable in the clause.
				$b)$ To complete the instance, green boxes are added to cover all the regions with depth 1. The green rectangles are forced to be in any kernel by making them large enough so that each covers an exclusive region.}
			\label{fig:clausevarsample}
		\end{minipage}
	\end{figure}
	
	\paragraph{Completing the instance.}
	Each rectangle in a variable or clause gadget, as described, has a region that no other rectangle covers (i.e., of depth 1).
	Thus, the coverage kernel of the instances described up to here is trivial: all the rectangles.
	To avoid this, we cover all the regions of depth 1 with \emph{green} rectangles (as illustrated in \Cref{fig:clausevarsample}.$b$) which are forced to be in any coverage kernel%
	\footnote{
		For simplicity, these green rectangles were omitted in \Cref{fig:clausesample-full} and \Cref{fig:clauseevaluation}.
	}. 
	For every clause gadget we add $11$ such green rectangles, and for each variable gadget for a variable $v$ occurring in $c_v$ clauses we add $3c_v + 2$ green rectangles. 
	
	Let $\varphi$ be a formula with $n$ variables and $m$ clauses.
	The instance of $k$-\textsc{Coverage Kernel} that we create for $\varphi$ has a total of $41m + 4n$ rectangles:
	($i.$) each clause gadget has 9 rectangles for the comb, and 11 green rectangles, for a total of $20m$ rectangles over all the clauses;
	($ii.$) a gadget for a variable $v$ has $4c_v + 2$ red and blue rectangles, and we add a green rectangle for each of those that does not connect to a clause gadget ($3c_v + 2$ per variable), thus adding a total of $7c_v + 4$  rectangles by gadget; and ($iii.$) over all variables, we add a total of $\sum_{i=1}^{n}{7c_{v_i} + 4} = 7(3m) + 4n = 21m+4n$ rectangles%
	\footnote{
		Note that $\sum_{i=1}^{n}{c_{v_i}} = 3m$ since exactly 3 variables occurs in each clause.
	}.

	\paragraph{Intuition: from minimum kernels to boolean values.}
	Consider a gadget for a variable $v$. Any minimum coverage kernel of the gadget is composed of all its green rectangles together with either all its blue or all its red rectangles. Thus the minimum number of rectangles needed to cover all the variable gadgets is fixed, and known. If all the red rectangles are present in the kernel, we consider that $v=1$, otherwise if all the blue rectangles are present, we consider that $v=0$ (see \Cref{fig:clauseevaluation} for an example).
	\begin{figure}[h]
		\centering	
		\begin{minipage}{.9\textwidth}
			\centering	
			\hspace*{-.25in}
			\includegraphics[scale=.75,page=12]{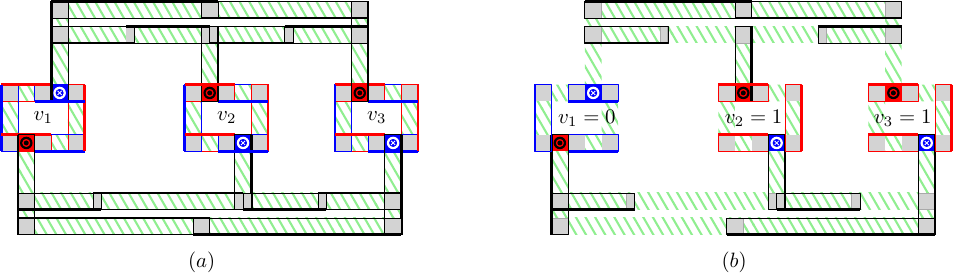}		
			\caption{$a)$ An instance for $\varphi = (\overline{v_1} \lor v_2 \lor v_3) \land (v_1 \lor \overline{v_2} \lor \overline{v_3})$ (the area covered by green rectangles is highlighted with green falling lines).
				$b)$ A cover of the gadgets corresponding to the assignment $v_1=0, v_2=1, v_3=1$ that does not satisfy $\varphi$.}
			\label{fig:clauseevaluation}
		\end{minipage}
	\end{figure}
	In the same way that choosing a value for $v$ may affect the output of a clause $C$ in which $v$ occurs, choosing a color to cover the gadget for $v$ may affect the number of rectangles required to cover the gadget for $C$.
	For instance, consider that the gadget for $v$ is covered with blue rectangles (i.e., $v=0$), and that $v$ occurs unnegated in $C$ (see the gadgets for $v_1$ and the second clause in \Cref{fig:clauseevaluation}).
	The respective leg of the gadget for $C$ meets $v$ in one of its red rectangles.
	Since red rectangles were not selected to cover the variable gadget, that leg is forced to cover the connection region shared with the variable gadget, and thus is forced to be in any kernel of the gadget for $C$.
	This corresponds to the fact that the literal of $v$ in $C$ evaluates to 0.
	If the same happens for the three variables in $C$ (i.e., $C$ is not satisfied by the assignment), then to cover its gadget  at least six of the black rectangles will be required (see the gadget for the second clause in \Cref{fig:clauseevaluation}).
	However, if at least one of its legs can be disposed of (i.e., at least one of the literals evaluates to 1), then the clause gadget can be covered with five of its black rectangles (see the gadget for the first clause in \Cref{fig:clauseevaluation}).
	The minimum number of rectangles needed to cover all the variable gadgets is fixed and known: all the green rectangles, and one half of the red/blue rectangles of each variable. 
	Therefore, it suffices to show that there is an assignment satisfying a 3-SAT formula if and only if every clause gadget can be covered by five of its black rectangles (plus all its green rectangles).
	
	\paragraph{Reduction.} 
	We prove the theorem in two steps. First, we show that such an instance has a coverage kernel of size $31m + 3n$ if and only $\varphi$ is satisfiable. Therefore, answering \textsc{$k$-Coverage-Kernel} over this instance with $k=31m+3n$ yields an answer for \textsc{Planar 3-SAT} on $\varphi$. Finally, we will show that the instance described matches all the restrictions in \Cref{theo:coverage_hardness}, under minor variations.
	
	\paragraph{($\Rightarrow$)}
	Let $\varphi$ be 3-CNF formula with $n$ variables $v_1, \ldots, v_n$ and $m$ clauses $C_1, \ldots, C_m$, let $\mathcal{B}$ be a set of boxes created as described above for $\varphi$, and let $\{v_1 = \alpha_1, \ldots, v_n = \alpha_n\}$ be an assignment which satisfies $\varphi$. We create a coverage kernel \set{K} of \set{B} as follows:
	\begin{itemize}
		\item  For each variable gadget for $v_i \in \{v_1, \ldots, v_n\}$ such that $\alpha_i = 0$ (resp. $\alpha_i = 1$), add to \set{K} all but its red (resp. blue) rectangles, thus covering the entire gadget minus its red (resp. blue) connection regions.
		This uncovered regions, which must connect with clauses in which the literal of $v_i$ evaluates to 0, will be covered later with the legs of the clause gadgets.
		Over all variables, we add to \set{K} a total of $\sum_{i=1}^{n}{\left( (7-2)c_{v_i} + (4-1)\right)} = 5(3m) + 3n = 15m+3n$ rectangles.%
		
		\item For each clause $C_i \in \{C_1, \ldots, C_n\}$, add to \set{K} all its green rectangles and the legs that connect with connection regions of the variable gadgets left uncovered in the previous step.
		Note that at least one of the legs of $C_i$  is not added to \set{K} since at least one of the literals in the clause evaluates to 1, and  the connection region corresponding to that literal is already covered by the variable gadget.
		Thus, the redundant regions of the gadget for $C_i$ can be covered with five of its black rectangles (including the legs already added). 
		So, finally add to \set{K} black rectangles from the clause for $C_i$ as needed (up to a total of five), until all the redundant regions are covered (and with the green rectangles, the entire gadget).
		Over all clauses, we add a total of $(11+5)m = 16m$ rectangles.
	\end{itemize}
	By construction, \set{K} is a coverage kernel of \set{B}: it covers completely every variable and clause gadget.
	Moreover, the size of \set{K} is $(15m + 3n) + 16m = 31m + 3n$.
	
	\paragraph{($\Leftarrow$)}
	
	Let $\varphi$ be 3-CNF formula with $n$ variables $v_1, \ldots, v_n$ and $m$ clauses $C_1, \ldots, C_m$,
	let $\mathcal{B}$ be a set of boxes created as described above for $\varphi$, 
	and let \set{K} be a coverage kernel of \set{B} whose size is $(15m + 3n) + 16m = 31m + 3n$.
	Any coverage kernel of \set{B} must include all its green rectangles. Furthermore, to cover any clause gadget at least five of its black rectangles are required, and to cover any variable rectangle, at least half of its red/blue rectangles are required.
	Thus, any coverage kernel of \set{B} most have at least $(11+5)m = 16m$ of the clause rectangles, and at least  $\sum_{i=1}^{n}{\left( (3+4/2)c_{v_i} + (3)\right)} = 5(3m) + 3n = 15m+3n$ of the variable rectangles are required.
	Hence, \set{K} must be a coverage kernel of minimum size.  
	
	Since the redundant regions of any two gadgets are independent, \set{K} must cover each gadgets optimally  (in a local sense).
	Given that the intersection graph of the red/blue rectangles of a variable gadget is a ring (see \Cref{fig:variablegadget}.d for an illustration), the only way to cover a variable gadget optimally is by choosing either all its blue or all its red rectangles (together with the green rectangles).
	Hence, the way in which every variable gadget is covered is consistent with an assignment for its variable as described before in the intuition. 
	Moreover, the assignment induced by \set{K} must satisfy $\varphi$: in each clause gadget, at least one of the legs was discarded (to cover the gadget with 5 rectangles), and at least the literal in the clause corresponding to that leg evaluates to 1.
	
	\paragraph{Meeting the restrictions.}
	Now we prove that the instance of \textsc{Minimum Coverage Kernel} generated for the reduction meets the restrictions of the theorem.
	First, we show the bounded clique-number and vertex-degree properties for the intersection graph of a clause gadget and its three respective variable gadgets. 
	In	\Cref{fig:restrictions} we illustrate the intersection graph for the clause $\varphi = (\overline{v_1} \lor v_2 \lor v_3)$. The sign of the variables in the clause does not change the maximum vertex degree or the clique-number of the graph, so the figure is general enough for our purpose.
	Since we consider the rectangles composing the instance to be closed rectangles, if two rectangles containing at least one point in common (in their interior or boundary), their respective vertices in the intersection graph are adjacent.
	Note that in \Cref{fig:restrictions} the vertices with highest degree are the ones corresponding to legs of the clause gadget (rectangles 1, 2, and 3).
	There are 4-cliques in the graph, for instance the right lower corner of the green rectangle denoted $h$ is covered also by rectangles $1,4$ and $d$, and hence their respective vertices form a clique.
	However, since there is no point that is covered by five rectangles at the same time, there are no 5-cliques in the graph. 
	\begin{figure}[h]
		\centering	
		\begin{minipage}{.9\textwidth}
			\centering	
			\hspace*{-.25in}
			\includegraphics[scale=.8,page=13]{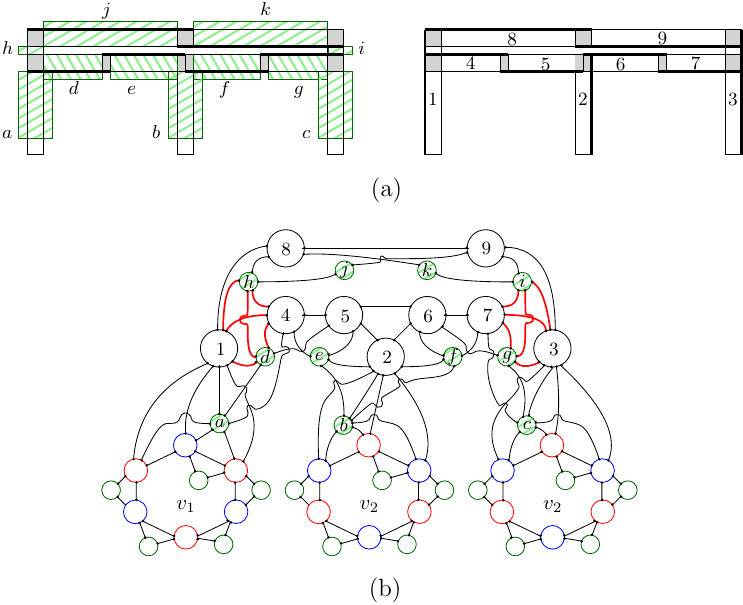}		
			\caption{The intersection graph of the instance of the \textsc{Minimum Coverage Kernel} problem corresponding to $\varphi = (\overline{v_1} \lor v_2 \lor v_3)$ (see the complete instance in \Cref{fig:clausevarsample}.b).
				$a)$ Numbering of the rectangles of the clause gadget.
				$b)$ The intersection graph of the instance: the vertices corresponding to the legs have degree 8, which is the highest; and the red fat edges highlight two 4-cliques.}
			\label{fig:restrictions}
		\end{minipage}
	\end{figure}
	Finally, note that, since the clause gadgets are located according to the planar embedding of the formula, they are pairwise independent.%
	\footnote{Two clause gadgets are pairwise independent if the rectangles composing them are pairwise independent, as well as the rectangles where they connect with their respective variables gadgets.}
	Thus, the bounds on the clique-number and vertex-degree of the intersection graph of any clause gadget extend also to the intersection graph of an entire general instance.
	\qed
\end{proof}

\subsection{Extension to Box Cover}

Since the \textsc{Minimum Coverage Kernel} problem is a special case of the \textsc{Box Cover} problem, the result of \Cref{theo:coverage_hardness} also applies to the \textsc{Box Cover} problem.
However, in \Cref{theo:coverage_hardness_boxcover} we show that this problem remains hard under even more restricted settings. 

\begin{theorem}\label{theo:coverage_hardness_boxcover}
	Let $P$, $\mathcal{B}$ be a set of $m$ points and $n$ boxes in the plane, respectively,  and let $G$ be the intersection graph of $\mathcal{B}$.
	Solving \textsc{Box Cover} over $\mathcal{B}$ and $P$ is \emph{\textsf{NP}}-complete even 
	if every point in $P$ is covered by at most two boxes of $\mathcal{B}$, and
	$G$ is planar, has clique-number at most 2, and vertex-degree at most 4.
\end{theorem}

\begin{proof}
	We use the same reduction from \textsc{Planar 3-SAT}, but with three main variations in the gadgets: we drop the green rectangles of both the variable and clause gadgets, add points within the redundant and connection region of both variable and clause gadgets, and separate the rectangles numbered 5 and 6 of each clause gadget so they do not intersect (see \Cref{fig:boxcoverinstance} for an example).
	\begin{figure}[h]
		\centering	
		\begin{minipage}{.9\textwidth}
			\centering	
			\hspace*{-.25in}
			\includegraphics[scale=.8,page=14]{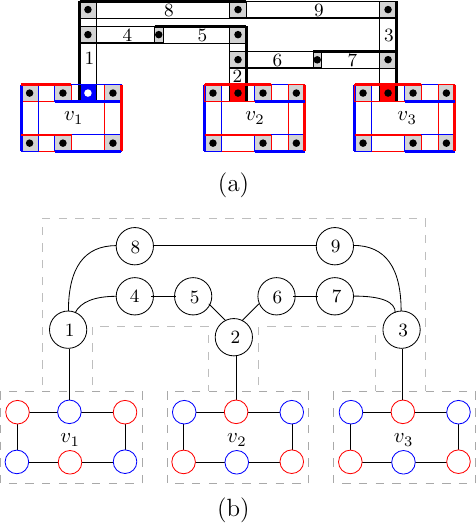}		
			\caption{
				$a)$ The instance of the \textsc{Box Cover} problem corresponding to $\varphi = (\overline{v_1} \lor v_2 \lor v_3)$.
				$b)$ The intersection graph of the instance: the vertices corresponding to the legs have degree 3, which is the highest; there are no 3-cliques in the graph; and the graph is planar and can be drawn within the planar embedding of $\varphi$ (highlighted with dashed lines). }
			\label{fig:boxcoverinstance}
		\end{minipage}
	\end{figure}
	Since the interior of every connection or redundant region is covered by at most two of the rectangles in the gadgets, every point of the instance we create is contained in at most two boxes.
	In	\Cref{fig:boxcoverinstance}.b we illustrate the intersection graph for the clause $\varphi = (\overline{v_1} \lor v_2 \lor v_3)$. 
	Since the sign of the variables in the clause does not change the maximum vertex degree or the clique-number of the graph, or its planarity, the properties we mention next are also true for any clause.
	Note that three is the maximum vertex-degree of the intersection graph, and  that there are no 3-cliques.
	Also note that the intersection graph can be drawn within the planar embedding of $\varphi$ so that no two edges cross, and hence the graph is planar.
	Again, due to the pairwise independence of the clause gadgets, these properties extend to the entire intersection graph of a general instance.
	\qed
\end{proof}

In the next section, we complement these hardness results with two approximation algorithms for the \textsc{Minimum Coverage Kernel} problem.

\section{Efficient approximation of Minimum Coverage Kernels}\label{sec:covkernels_approximation}
Let $\mathcal{B}$ be a set of $n$ boxes in $\mathbb{R}^{d}$, and let $\mathcal{D(B)}$ be a coverage discretization of $\mathcal{B}$ (as defined in \Cref{sec:background}).
A \emph{weight index} for $\mathcal{D(B)}$ is a data structure which can perform the following operations: 
\begin{itemize}
	\item  \emph{Initialization:} Assign an initial unitary weight to every point in $\mathcal{D(B)}$;
	\item \emph{Query:} Given a box $b \in \mathcal{B}$, find the total weight of the points in $b$.
	\item \emph{Update:} Given a box $b \in \mathcal{B}$, multiply the weights of all the points within $b$ by a given value $\alpha \ge 0$;
\end{itemize}
We assume that the weights are small enough so that arithmetic operations over the weights can be performed in constant time. 
There is a trivial implementation of a weight index with initialization and update time within $\bigo{n^{d}}$, and with constant query time. 
In this section we describe an efficient implementation of a weight index, and combine this data structure  with two existing approximation algorithms for the \textsc{Box Cover} problem~\cite{Lovasz75,BronnimannG95} and obtain improved approximation algorithms (in the running time sense) for the \textsc{Minimum Coverage Kernel} problem.

\subsection{An Efficient Weight Index for a Set of Boxes}
We describe a weight index for $\mathcal{D(B)}$ which can be initialized in time within $\bigo{n^\frac{d+1}{2}}$, and with query and update time within $\bigo{n^\frac{d-1}{2}\log n}$.
Let us consider first the case of a set $I$ of $n$ intervals.

\paragraph{\textbf{A weight index for a set of intervals}.}
A trivial weight index which explicitly saves the weights of each point in $\mathcal{D}(I)$ can be initialized in time within $\bigo{n \log n}$, has linear update time, and constant query time.
We show that by sacrificing query time (by a factor within $\bigo{\log n}$) one can improve update time to within $\bigo{\log n}$.
The main idea is to maintain the weights of each point of  $\mathcal{D}(I)$ indirectly using a tree.

Consider a balanced binary tree whose leafs are in one-to-one correspondence with the values in $\mathcal{D}(I)$ (from left to right in a non-decreasing order).
Let $p_v$ denote the point corresponding to a leaf node $v$ of the tree.
In order to represent the weights of the points in $\mathcal{D}(I)$, we store a value $\mu(v)$ at each node $v$ of the tree subject to the following invariant:
for each leaf $v$, the weight of the point $p_v$ equals the product of the values $\mu(u)$ of all the ancestors $u$ of $v$ (including $v$ itself).
The $\mu$ values allow to increase the weights of many points with only a few changes. 
For instance, if we want to double the weights of all the points we simply multiply by 2 the value $\mu(r)$ of the root $r$ of the tree. 
Besides the $\mu$ values, to allow efficient query time we also store at each node $v$ three values $min(v),max(v),\omega(v)$:
the values  $min(v)$ and $ max(v)$ are the minimum and maximum $p_u$, respectively, such that $u$ is a leaf of the tree rooted at $v$; 
the value $\omega(v)$ is the sum of the weights of all $p_u$ such that $u$ is a leaf of the tree rooted at $v$.

Initially, all the $\mu$ values are set to one.
Besides, for every leaf $l$ of the tree $\omega(l)$ is set to one, while  $min(l)$ and $max(l)$ are set to $p_l$. The $min$, $max$ and $\omega$ values of every internal node $v$ with children $l,r$,  are initialized in a bottom-up fashion as follows:
 $min(v) = \min\{min(l), min(r)\}$;
 $max(v) = \max\{max(l), max(r)\}$;
 $\omega(v)=\mu(v)\cdot\left(\omega(l) + \omega(r)\right)$.
It is simple to verify that after this initialization, the tree meets all the invariants mentioned above.
We show in \Cref{theo:windex_intervals} that this tree can be used as a weight index for $\mathcal{D}(I)$.

\begin{theorem}~\label{theo:windex_intervals}
	Let $I$	be a set of $n$ intervals in $\mathbb{R}$. There exists a weight index for $\mathcal{D}(I)$ which can be initialized in time within $\bigo{n \log n}$, and with query and update time within $\bigo{\log n}$.
\end{theorem}

\begin{proof}
	Since intervals have linear union complexity, $\mathcal{D}(I)$ has within $\bigo{n}$ points, and it can be computed in linear time after sorting, for a total time within $\bigo{n \log n}$.
	We store the points in the tree described above. Its initialization can be done in linear time since the tree has within $\bigo{n}$ nodes, and when implemented in a bottom-up fashion, the initialization of the $\mu, \omega, min,$ and $max$ values, respectively, cost constant time per node. 
	
	To 	analyze the query time, let $\textsf{totalWeight}(a,b,t)$ denote the procedure which finds the total weight of the points corresponding to leafs of the tree rooted at $t$ that are in the interval $[a,b]$.
	This procedure can be implemented as follows:
	
	\begin{enumerate}
		\item if $[a,b]$ is disjoint to $[min(t),max(t)]$ return 0;
		\item if $[a,b]$ completely contains $[min(t),max(t)]$ return $\omega(r)$;
		\item if both conditions fail (leafs must meet either 1. or 2.), let $l,r$ be the left and right child of $t$, respectively;
		\item if $a > max(l)$ return $\mu(t) \cdot \textsf{totalWeight}(a,b,r)$;
		\item if $b < min(r)$ return $\mu(t) \cdot \textsf{totalWeight}(a,b,l)$;
		\item otherwise return $\mu(t) (\textsf{totalWeight}(a, \infty, l) + \textsf{totalWeight}(-\infty,b,r))$.
	\end{enumerate}
	
	\noindent
	Due to the invariants to which the $min$ and $max$ values are subjected, every leaf $l$ of $t$ corresponding to a point in $[a,b]$ has an ancestor (including $l$ itself) which is visited during the call to \textsf{totalWeight} and which meets the condition in step 2.
	For this, and because of the invariants to which the $\omega$ and $\mu$ values are subjected, the procedure $\textsf{totalWeight}$ is correct.
	Note that the number of nodes visited is at most 4 times the height $h$ of the tree:
	when both children need to be visited, one of the endpoints of the interval to query is replaced by $\pm \infty$,
	which ensures that in subsequent calls at least one of the children is completely covered by the query interval.
	Since $h \in \bigo{\log n}$, and the operations at each node consume constant time, the running time of  \textsf{totalWeight} is within $\bigo{\log n}$.
	
	Similarly, to analyze the update time, let $\textsf{updateWeights}(a,b,t, \alpha)$ denote the procedure which multiplies by a value $\alpha$ the weights of the points in the interval $[a,b]$ stored in leafs descending from $t$.
	This can be implemented as follows:
	
	\begin{enumerate}
		\item if $[a,b]$ is disjoint to $[min(t),max(t)]$, finish;
		\item if $[a,b]$ completely contains $[min(t),max(t)]$ set $\mu(r)=\alpha \cdot \mu(r)$, set $ \omega(r)=\alpha \cdot \omega(r)$,  and finish;
		\item if both conditions fail, let $l,r$ be the left and right child of $t$, respectively;
		\item if $a > max(l)$, call $\textsf{updateWeights}(a,b,r, \alpha)$;
		\item else if $b < min(r)$, call $\textsf{updateWeights}(a,b,l, \alpha)$;
		\item otherwise,  call $\textsf{updateWeights}(a, \infty, l, \alpha)$, and $\textsf{updateWeights}(-\infty,b,r, \alpha)$;
		\item finally, after the recursive calls set $\omega(t)=\mu(t) \cdot \left(\omega(l) + \omega(r)\right)$, and finish.
	\end{enumerate}

	\noindent
	Note that, for every point $p_v$ in $[a,b]$ corresponding to a leaf $v$ descending from $t$, the $\mu$ value of exactly one of the ancestors of $u$ changes (by a factor of $\alpha$):
	at least one changes because of the invariants to which the $min$ and $max$ values are subjected (as analyzed for \textsf{totalWeight}); and no more than one can change because once $\mu$ is assigned for the first time to some ancestor $u$ of $v$, the procedure finishes leaving the descendants of $v$ untouched. 
	The analysis of the running time is analogous to that of \textsf{totalWeight}, and thus within $\bigo{\log n}$.
	\qed
\end{proof}


The weight index for set of intervals described in \Cref{theo:windex_intervals} plays an important role in obtaining an index for a higher dimensional set of boxes. In a step towards that, we first describe how to use one dimensional indexes to obtain indexes for another special case of sets of boxes, this time in high dimension.

\paragraph{\textbf{A weight index for a set of slabs}.}
A box $b$ is said to be a slab within another box $\Gamma$ if $b$ covers completely $\Gamma$ in all but one dimension (see \Cref{fig:slabs}.a for an illustration). 
\begin{figure}[t]
	\centering
	\begin{minipage}{.9\textwidth}
		\centering	
		\includegraphics[scale=.8]{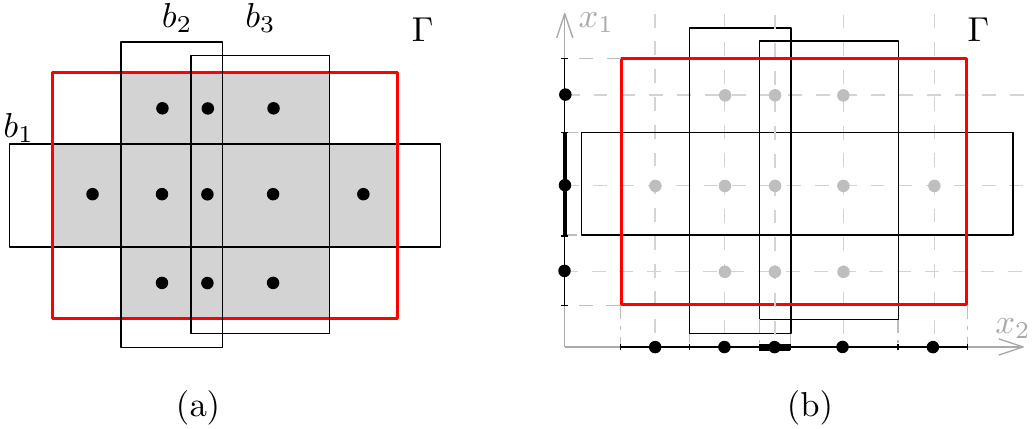}
		\caption{(a) An illustration in the plane of three boxes equivalent to slabs when restricted to the box $\Gamma$.
			The dots correspond to a set of points which discretize the (grayed) region within $\Gamma$ covered by the slabs.
		}\label{fig:slabs}
	\end{minipage}
\end{figure}
Let $\mathcal{B}$ be a set of $n$ $d$-dimensional boxes that are slabs within another box $d$-dimensional box $\Gamma$.
Let $\mathcal{B} |_\Gamma$ denote the set $\{b \cap \Gamma \mid b \in \mathcal{B}\}$ of the boxes $\in  \mathcal{B}$ restricted to $\Gamma$.
We describe a weight index for $\mathcal{D}(\mathcal{B} |_\Gamma)$ with initialization time within the size $n$, and with update and query time within $\bigo{n \log n}$.

For all $i = [1..d]$, let $\mathcal{B}_i$ be the subset of slabs that are orthogonal to the $i$-th dimension, and let $I_i$ be the set of intervals resulting from projecting $\Gamma$ and each rectangle in $\mathcal{B} |_\Gamma$ to the $i$-th dimension (see \Cref{fig:slabs}.b for an illustration).
The key to obtain an efficient weight index for a set of slabs is the fact that weight indexes for $\mathcal{D}(I_1), \ldots \mathcal{D}(I_d)$ can be combined without much extra computational effort into a weight index for $\mathcal{D}(\mathcal{B} |_\Gamma)$.
Let $p$ be a point $\mathcal{D}(\mathcal{B} |_\Gamma)$ and let $x_i(p)$ denote the value of the $i$-th coordinate of $p$. Observe that for all $i \in [1..d]$, $x_i(p) \in \mathcal{D}(I_i)$ (see \Cref{fig:slabs}.b for an illustration).
This allows the representation of the weight of each point $p \in \mathcal{D(B)}$ by means of the weights of $x_i(p)$ for all $i \in [1..d]$. We do this by maintaining the following \emph{weight invariant}: the weight of a point $p \in \mathcal{D}(\mathcal{B} |_\Gamma)$ is equal to $\prod_{i=1}^{d}{\left(\text{weight of }x_i(p) \text{ in } \mathcal{D}(I_i)\right)}$.
\begin{lemma}~\label{lem:windex_slabs}
	Let $\mathcal{B}$ be a set of $n$ $d$-dimensional boxes that are equivalent to slabs when restricted to another $d$-dimensional box $\Gamma$. There exists a weight index for $\mathcal{D(B)}$ which can be initialized in time within $\bigo{n \log n}$, and with query and update time within $\bigo{\log n}$.
\end{lemma}
\begin{proof}
	Let $\mathcal{B}_i$ be the subset of $\mathcal{B}|_\Gamma$ orthogonal to the $i$-th dimension, and let $I_i$ be the set of intervals resulting from projecting $\Gamma$ and each rectangle in $\mathcal{B}_i$ to the $i$-th dimension.
	Initialize a weight index for $\mathcal{D}(I_i)$ as in \Cref{theo:windex_intervals}, for all $i \in [1..d]$. 
	Since the weights of all the points in the one dimensional indexes are initialized to one, the weight of every point in $\mathcal{D}(\mathcal{B}|_\Gamma)$ is also initialized to one, according to the weight invariant. This initialization can be done in total time within $\bigo{n \log n}$.
	
	Let $b \in \mathcal{B}|_\Gamma$  be a box which covers $\Gamma$ in every dimension except for the $i$-th one, for some $i \in [1..d]$ (i.e., $b \in B_i$), and let $P_i$ bet the subset of $\mathcal{D}(I_i)$ contained within the projection of $b$ to the $i$-th dimension.
	The set of points of $\mathcal{D}(\mathcal{B}|_\Gamma)$  that are within $b$ can be generated by the expression
	$\{ (a_1, \ldots, a_d) \mid a_1 \in  \mathcal{D}(I_1) \land \ldots \land a_{i-1} \in \mathcal{D}(I_{i-1}) \land a_i \in P_i \land a_{i+1} \in  \mathcal{D}(I_{i+1}) \land \ldots \land a_d \in \mathcal{D}(I_d) \}$.
	Therefore, the total weight of the points within $b$ is given by the total weight of the points $x_i(p)$ for all $p \in P_i$ multiplied by the total weight of the points in $ \mathcal{D}(I_j)$, for all $j=[1..d]$ distinct from $i$.
	
	To query the total weight of the points of $\mathcal{D}(\mathcal{B} |_\Gamma)$ within a box $b \in B_i$ we query the weight index of $\mathcal{D}(I_i)$ to find the total weight of the points in the projection of $b$ to the $i$-th dimension (in time within $\bigo{\log n}$), then query the remaining $d-1$ indexes to find the total weight store in the index (stored at the $\omega$ value of the root of the trees), and return the product of those $d$ values. Clearly the running time is within $\bigo{\log n}$.
	
	The update is similar: to multiply by a value $\alpha$ the weight of all the points of $\mathcal{D}(\mathcal{B} |_\Gamma)$ within a box $b \in B_i$ we simply update the weight index of  $\mathcal{D}(I_i)$ multiplying by $\alpha$ all the weights of the points within the projection of $b$ to the $i$-th dimension, and leave the other $d-1$ weight indexes untouched.
	The running time of this operation is also within $\bigo{n \log n}$, and the invariant remains valid after the update.
	\qed	
\end{proof}

\Cref{lem:windex_slabs} shows that there are weight indexes for a set of slabs $\mathcal{B}$ within another box $\Gamma$ that significantly improve the approach of explicitly constructing $\mathcal{D}(\mathcal{B}|_\Gamma)$, 
an improvement that grows exponentially with the dimension $d$.
We take advantage of this to describe a similar improvement for the general case.


\paragraph{\textbf{A Weight Index for The General Case}.}
We now show how to maintain a weight index of a general set $\mathcal{B}$ of $d$-dimensional boxes.
The main idea is to partition the space into cells such that, within each cell $c$, any box $b \in \mathcal{B}$ either completely contains $c$ or is equivalent to a \emph{slab}.
Then, we use weight indexes for slabs (as described in \Cref{lem:windex_slabs})  to index the weights within each of the cells.
This approach was first introduced by \textcite{Overmars1991} in order to compute the volume of the region covered by a set boxes, and similar variants were used since then to compute other measures~\cite{2017-COCOON-DepthDistributionInHighDimension-BabrayPerezRojas,Chan2013,YildizHershbergerSuri11}.
The following lemma summarizes the key properties of the partition we use:

\begin{lemma}[Lemma 4.2 of \textcite{Overmars1991}]~\label{lem:oyap}
	Let $\mathcal{B}$ be a set of $n$ boxes in $d$-dimensional space. There exist a binary partition tree for storing any subset of $\mathcal{B}$ such that
	\begin{itemize}
		\item It can be computed in time within $\bigo{n^{d/2}}$, and it has $\bigo{n^{\frac{d}{2}}}$ nodes;
		\item Each box is stored in $\bigo{n^{\frac{d-1}{2}}}$ leafs;
		\item The boxes stored in a leaf are slabs within the cell corresponding to the node;
		\item Each leaf stores no more than $\bigo{\sqrt{n}}$ boxes.
	\end{itemize}
\end{lemma}

Consider the tree of \Cref{lem:oyap}.
Analogously to the case of intervals, we augment this tree with information to support the operations of a weight index efficiently. 
At every node $v$ we store two values $\mu(v), \omega(v)$:  the first allows to multiply all the weights of the points of $\mathcal{D(B)}$ that are maintained in leafs descending from $v$ (allowing to support updates efficiently);
while $\omega(v)$ stores the total weight of these points (allowing to support queries efficiently).
To ensure that all/only the nodes that intersect a box $b$ are  visited during a query or update operation, we store at each node the boundaries of the cell corresponding to that node.
Furthermore, at every leaf node $l$ we implicitly represent the points of $\mathcal{D(B)}$ that are inside the cell corresponding to $l$ using a weight index for slabs.

To initialize this data structure, all the $\mu$ values are set to one. Then the weight index within each leaf cell are initialized. Finally, the $\omega$ values of every node $v$ with children $l,r$,  are initialized in a bottom-up fashion setting
$\omega(v)=\mu(v)\cdot\left(\omega(l) + \omega(r)\right)$.
We show in \Cref{theo:windex_gcase} how to implement the weight index operations over this tree and we analyze its running times.

\begin{theorem}\label{theo:windex_gcase}
	Let $\mathcal{B}$ be a set of $n$ $d$-dimensional boxes. There is a weight index for $\mathcal{D(B)}$ which can be initialized in time within $\bigo{n^\frac{d+1}{2}}$, and with query and update time within $\bigo{n^\frac{d-1}{2}\log n}$.
\end{theorem}

\begin{proof}
	The initialization of the index, when implemented as described before, runs in constant time for each internal node of the tree, and in time within $\bigo{\sqrt{n} \log n}$ for each leaf (due to \Cref{lem:windex_slabs}, and to the last item of \Cref{lem:oyap}). Since the tree has $\bigo{n^{\frac{d}{2}}}$ nodes, the total running time of the initialization is within $\bigo{n^{\frac{d+1}{2}}\log n}$.
	
	Since the implementations of the query and update operations are analogous to those for the intervals weight index (see the proof of \Cref{theo:windex_intervals}), we omit the details of their correctness. 
	While performing a query/update operation at most $\bigo{n^{\frac{d-1}{2}}}$ leafs are visited (due to the third item of \Cref{lem:oyap}), and since the height of the tree is within $\bigo{\log n}$, at most $\bigo{n^{\frac{d-1}{2}}\log n}$ internal nodes are visited in total.
	Hence, the cost of a query/update operation within each leaf is within $\bigo{\log n}$ (by \Cref{lem:windex_slabs}), and is constant within each internal node. 
	Thus, the total running of a query/update operation is within $\bigo{n^\frac{d-1}{2}\log n}$.
	\qed
\end{proof}

\subsection{Practical approximation algorithms for Minimum Coverage Kernel.}
Approximating the \textsc{Minimum Coverage Kernel} of a set $\mathcal{B}$ of boxes via approximation algorithms for the \textsc{Box Cover} problem requires that $\mathcal{D(B)}$ is explicitly constructed.
However, the weight index described in the proof of \Cref{theo:windex_gcase} can be used to significantly improve the running time of these algorithms. We describe below two examples.

The first algorithm we consider is the greedy $\bigo{\log n}$-approximation algorithm by \textcite{Lovasz75}. 
The greedy strategy applies naturally to the \textsc{Minimum Coverage Kernel} problem: iteratively pick the box which covers the most yet uncovered points of $\mathcal{D(B)}$, until there are no points of $\mathcal{D(B)}$ left to cover. 
To avoid the explicit construction of $\mathcal{D(B)}$ three operations most be simulated:
($i$.) find how many uncovered points are within a given a box $b \in \mathcal{B}$; 
($ii$.) delete the points that are covered by a box $b \in \mathcal{B}$; and
($iii$.) find whether a subset $\mathcal{B'}$ of  $\mathcal{B}$ covers all the points of $\mathcal{D(B)}$.

For the first two we use the weight index described in the proof of \Cref{theo:windex_gcase}:
to delete the points within a given box $b \in \mathcal{B}$ we simply multiply the weights of all the points of $\mathcal{D(B)}$ within $b$ by $\alpha = 0$;
and finding the number of uncovered points within a box $b$ is equivalent to finding the total weight of the points of $\mathcal{D(B)}$ within $b$.
For the last of the three operations we use the following observation:
\begin{observation}
	Let $\mathcal{B}$ be a set of $d$-dimensional boxes, and let $\mathcal{B'}$ be a subset of $\mathcal{B}$. The volume of the region covered by $\mathcal{B'}$ equals that of $\mathcal{B}$ if and only if $\mathcal{B'}$ and $\mathcal{B}$ cover the exact same region.
\end{observation}
%
Let $\mathtt{OPT}$ denote the size of a minimum coverage kernel of $\mathcal{B}$, and let $N$ denote the size of $\mathcal{D(B)}$ ($N \in \bigo{n^{d}}$). The greedy algorithm of \textcite{Lovasz75}, when run over the sets $\mathcal{B}$ and $\mathcal{D(B)}$ works in $\bigo{\mathtt{OPT} \log N}$ steps; and at each stage a box is added to the solution. The size of the output is within $\bigo{\mathtt{OPT} \log N} \subseteq \bigo{\mathtt{OPT} \log n}$. This algorithm can be modified to achieve the following running time, while achieving the same approximation ratio:

\begin{theorem}\label{theo:approx_logn}
	Let $\mathcal{B}$ be a set of $n$ boxes in $\mathbb{R}^{d}$ with a minimum coverage kernel of size $\mathtt{OPT}$. Then, a \textsc{Coverage Kernel} of $\mathcal{B}$ of size within $\bigo{\mathtt{OPT} \log n}$ can be computed in time within $\bigo{\mathtt{OPT} \cdot n^{\frac{d}{2} + 1}\log^2 n}$.
\end{theorem}
\begin{proof}
	We initialize a weight index as in \Cref{theo:windex_gcase}, which can be done in time $\bigo{n^\frac{d+1}{2}}$, and compute the volume of the region covered by $\mathcal{B}$, which can be done in time within $\bigo{n^{d/2}}$~\cite{Chan2013}. Let $C$ be an empty set.
	At each stage of the algorithm, for every box $b \in \mathcal{B}\setminus C$ we compute the total weight of the points inside $b$ (which can be done in time within $n^{\frac{d-1}{2}}\log n$ using the weight index). 
	We add to $C$ the box with the highest total weight, and update the weights of all the points within this box to zero (by multiplying their weights by $\alpha=0$) in time within $n^{\frac{d-1}{2}}\log n$. 
	If the volume of the region covered by $C$ (which can be computed in $\bigo{n^{d/2}}$-time~\cite{Chan2013}) is the same as that of $\mathcal{B}$, then we stop and return $C$ as the approximated solution. 
	The total running time of each stage is within $\bigo{n^{\frac{d+1}{2}}\log n}$. This, and the fact that the number of stages is within $\bigo{\mathtt{OPT} \log n}$ yield the result of the theorem.
	\qed	
\end{proof}

Now, we show how to improve \citeauthor{BronnimannG95}'s $\bigo{\log \mathtt{OPT}}$ approximation algorithm~\cite{BronnimannG95} via a weight index. First, we describe their main idea.
Let $w : \mathcal{D(B) \rightarrow \mathbb{R}}$ be a weight function for the points of $\mathcal{D(B)}$, and for a subset $\mathcal{P} \subseteq \mathcal{D(B)}$ let $w(P)$ denote the total weight of the points in $\mathcal{P}$.
A point $p$ is said to be $\varepsilon$-heavy, for a value $\varepsilon \in (0,1]$, if $w(p) \ge \varepsilon w(\mathcal{D(B)})$, and $\varepsilon$-light otherwise.
A subset $\mathcal{B}' \subseteq \mathcal{B}$ is said to be an $\varepsilon$-net with respect to $w$ if for every $\varepsilon$-heavy point $p \in \mathcal{D(B)}$ there is a box in $\mathcal{B}' $ which contains  $p$.
Let $\mathtt{OPT}$ denote the size of a minimum coverage kernel of $\mathcal{B}$, and let $k$ be an integer such that $k/2 \le \mathtt{OPT} < k$. 
The algorithm initializes the weight of each point in $\mathcal{D(B)}$ to 1, and repeats the following \emph{weight-doubling step} until every range is $\frac{1}{2k}$-heavy: find a $\frac{1}{2k}$-light point $p$ and double the weights of all the points within every box $b \in \mathcal{B}$. 
When this process stops, it returns a $\frac{1}{2k}$-net $C$ with respect to the final weights as the approximated solution.

Since each point in $\mathcal{D(B)}$ is $\frac{1}{2k}$-heavy, $C$ covers all the points of $\mathcal{D(B)}$.
Hence, if a $\frac{1}{2k}$-net of size $\bigo{kg(k)}$ can be computed efficiently, this algorithm computes a solution of size $\bigo{kg(k)}$. Besides, \textcite{BronnimannG95} showed that for a given $k$, if more than $\mu_k = 4k \log (n/k)$ weight-doubling steps are performed, then $\mathtt{OPT} > 2k$.
This allows to guess the correct $k$ via exponential search, and to bound the maximum weight of any point by $n^4/k^3$ (which allows to represent the weights with $\bigo{\log n}$ bits).
See \citeauthor{BronnimannG95}'s article~\cite{BronnimannG95} for the complete details of their approach.

We simulate the operations over the weights of $\mathcal{D(B)}$ again using a weight index, this time with a minor variation to that of \Cref{theo:windex_gcase}: in every node of the space partition tree, besides the $\omega,\mu$ values, we  also store the minimum weight of the points within the cell corresponding to the node.
During the initialization and update operations of the weight index this value can be maintained as follows: for a node $v$ with children $l,r$, the minimum weight $min_\omega(v)$ of a point in the cell of $v$ can be computed as $min_\omega(v) = \omega(v) \cdot\min\{min_\omega(l), min_\omega(r)\}$.
This value allows to efficiently detect whether there are  $\frac{1}{2k}$-light points, and to find one in the case of existence by tracing down, in the partition tree,  the path from which that value comes.

To compute a $\frac{1}{2k}$-net, we choose a sample of $\mathcal{B}$ by performing at least $(16k \log 16k)$ random independent draws from $\mathcal{B}$. We then check whether it is effectively a $\frac{1}{2k}$-net, and if not, we repeat the process, up to a maximum of $\bigo{\log n}$ times. \textcite{HausslerW87} showed that such a sample is a $\frac{1}{2k}$-net with probability at least $1/2$. Thus, the expected number of samples needed to obtain a $\frac{1}{2k}$-net is constant, and since we repeat the process up to $\bigo{\log n}$ times, the probability of effectively finding one is at least $1 - \frac{1}{n^{\Omega(1)}}$.
We analyze the running time of this approach in the following theorem.

\begin{theorem}\label{theo:approx_randomized}
	Let $\mathcal{B}$ be a set of $n$ boxes in $\mathbb{R}^{d}$ with a minimum coverage kernel of size $\mathtt{OPT}$.
	A coverage kernel of $\mathcal{B}$ of size within $\bigo{\mathtt{OPT} \log \mathtt{OPT}}$ can be computed in $\bigo{\mathtt{OPT} n^{\frac{d+1}{2}} \log^2 n}$-expected time, with probability at least $1- \frac{1}{n^{\Omega(1)}}$.
\end{theorem}
\begin{proof}
	The algorithm performs several stages guessing the value of $k$. Within each stage we initialize a weight index in time within $\bigo{n^\frac{d+1}{2}}$. 
	Finding whether there is a $\frac{1}{2k}$-light  point can be done in constant time: the root of the partition tree stores both $w(\mathcal{D(B)})$ and the minimum weight of any point in the $\omega$ and $min_\omega$ values, respectively. 
	For every light point, the weight-doubling steps consume time within $\bigolr{n \times \left(n^\frac{d-1}{2}\log n\right)} \subseteq \bigo{n^{\frac{d+1}{2}}\log n}$ (by \Cref{theo:windex_gcase}). Since at each stage at most $4k \log (n/k)$ weight-doubling steps are performed, the total running time of each stage is within $\bigo{k n^{\frac{d+1}{2}}\log n \log \frac{n}{k}} \subseteq \bigo{k n^{\frac{d+1}{2}} \log^2 n}$.
	Given that $k$ increases geometrically while guessing its right value,  and since the running time of each stage is a polynomial function, the sum of the running times of all the stages is asymptotically dominated by that of the last stage, for which we have that $k \le \mathtt{OPT} \le 2k$. 
	Thus the result of the theorem follows.
	\qed	
\end{proof}

Compared to the algorithm of \Cref{theo:approx_logn}, this last approach obtains a better approximation factor on instances with small \textsc{Coverage Kernels} ($\bigo{\log n}$ vs. $\bigo{\log \mathtt{OPT}}$), but the improvement comes with a sacrifice, not only in the running time, but in the probability of finding such a good approximation.
In  two and three dimensions, weight indexes might also help to obtain practical $\bigo{ \log \log \mathtt{OPT}}$ approximation algorithms for the \textsc{Minimum Coverage Problem}.
We discuss this, and other future directions of research in the next section.

\section{Discussion}\label{sec:cover_discussion}
Whether it is possible to close the gap between the factors of approximation of \textsc{Box Cover} and \textsc{Orthogonal Polygon Covering} has been a long standing open question~\cite{KumarR03}.
The \textsc{Minimum Coverage Kernel} problem, intermediate between those two, has the potential of yielding answers in that direction, and has natural applications of its own~\cite{DalyLT16,LakshmananNWZJ02,PuM05}.
Trying to understand the differences in hardness between these problems we studied distinct restricted settings.
We show that while \textsc{Minimum Coverage Kernel} remains \textsf{NP}-hard under severely restricted settings, the  same can be said for the \textsc{Box Cover} problem under even more extreme settings;
and show that while the \textsc{Box Cover} and \textsc{Minimum Coverage Kernel} can be approximated by at least the same factors, the running time of obtaining some of those approximations can be significantly improved for the \textsc{Minimum Coverage Kernel} problem.

Another approach to understand what makes a problem hard is Parameterized Complexity~\cite{DowneyF99}, where the hardness of a problem is analyzed with respect to multiple parameters of the input, with the hope of finding measures gradually separating ``easy'' instances form the ``hard'' ones.
The hardness results described in \Cref{sec:covkernels_hardness} show that for the \textsc{Minimum Coverage Kernel} and \textsc{Box Cover} problems, the vertex-degree and clique-number of the underlaying graph are not good candidates of such kind of measures, opposed to what happens for other related problems~\cite{AlekseevBKL07}.

In two and three dimensions, the \textsc{Box Cover} problem can be approximated up to $\bigo{ \log \log \mathtt{OPT}}$~\cite{AronovES10}.
We do not know whether the running time of this algorithm can be also improved for the case of \textsc{Minimum Coverage Kernel} via a weight index.
We omit this analysis since the approach described in \Cref{sec:covkernels_approximation} is relevant when the dimension of the boxes is high (while still constant), as in distinct applications~\cite{DalyLT16,LakshmananNWZJ02,PuM05} of the \textsc{Minimum Coverage Kernel} problem.

\medskip 

%

\printbibliography

\begin{LONG}
	\input{appendix}
\end{LONG}

\end{document}